\documentclass[nofootinbib,10pt,twocolumn,preprintnumbers]{revtex4-1}
\usepackage{amsmath,amssymb,amsthm,appendix,bm,graphicx,float,array,multirow,rotfloat,caption,hyperref,cleveref,enumerate,mathdots,adjustbox,booktabs,parskip,mathtools,tikz,csquotes,rotating,empheq,braket}
\usetikzlibrary{arrows, positioning, decorations.pathmorphing, decorations.markings, decorations.pathreplacing, decorations.markings, matrix, patterns}
\usepackage{pgf,pgfarrows,pgfnodes,float,appendix, hyperref,slashed,breakurl,graphicx}
\definecolor{Color}{rgb}{0.28, 0.24, 0.55}
\definecolor{Orange}{rgb}{1,0.38,0.11}
\hypersetup{
    colorlinks = true,
    citecolor = Orange,
    linkcolor = Color,
    urlcolor  = Orange,
}
\definecolor{internationalorange}{rgb}{1.0, 0.31, 0.0}
\usepackage{graphicx}
\usepackage{subfigure}
\usepackage[margin=0.9in]{geometry}

\usepackage{enumerate}

\usepackage{amsmath}

\usepackage{color, colortbl}
\definecolor{Gray}{gray}{0.8}
\definecolor{GrayLight}{gray}{0.4}
\definecolor{Darkgreen}{RGB}{30,120,30}
\definecolor{granate}{rgb}{0.8039,0.2,0.2}
\newcommand{\beq}{\begin{equation}}
\newcommand{\eeq}{\end{equation}}
\newcommand{\bea}{\begin{eqnarray}}
\newcommand{\eea}{\end{eqnarray}}

\usepackage{eso-pic}

\theoremstyle{plain}
\newtheorem*{thm*}{Theorem}
\newtheorem{thm}{Theorem}[section]

\newtheorem{prop}[thm]{Proposition}

\theoremstyle{definition}
\newtheorem{defn}[thm]{Definition}
\newtheorem*{defn*}{Definition}

\begin{document}
\title{\Large {\bf{Emergent spacetime and the ergodic hierarchy}}}
\author{Elliott Gesteau}
\affiliation{Division of Physics, Mathematics and Astronomy, California Institute of Technology, Pasadena, CA 91125}

\begin{abstract}
Various diagnostics of the emergence of an arrow of time in the bulk description of a holographic theory have been proposed, including the decay of some real time correlation functions and the appearance of type III$_1$ von Neumann algebras carrying half-sided modular inclusions. This note puts forward a close parallel between these diagnostics and a quantum formulation of the ergodic hierarchy of dynamical systems. Theories with an emergent spacetime appear to sit near the top of this hierarchy.
\end{abstract}

\maketitle
\section{Introduction}

One of the main goals of quantum gravity is to understand the emergence of spacetime. While there exists a wide range of holographic theories, very few of them have a semiclassical limit that recovers weakly coupled quantum fields on a classical, smooth, gravitational spacetime in the dual description. It is therefore of central importance to understand the conditions required for a holographic theory to possess a ``good" semiclassical limit of this kind. In the context of AdS/CFT \cite{Maldacena:1997re}, theories with such a semiclassical limit in the bulk are the ones in which the 't Hooft coupling is strong and the gauge rank $N$ is taken to infinity. It makes sense to ask what is special about these theories. It has long been realized that a diagnostic of bulk locality is quantum chaos \cite{Shenker:2013pqa}, but coming up with a fully satisfactory definition of quantum chaos is still an open problem. While various diagnostics \cite{Bhattacharyya:2019txx} exist and capture different chaotic features of quantum systems, there is no clear overarching picture.

On the other hand, the theory of chaos for classical dynamical systems is under much better control. Thanks to pioneering work from the second part of the twentieth century, various diagnostics of classical chaos have been firmly established. One of the main upshots is that there is not a single way of diagnosing chaos in a classical dynamical system, but instead, a full \textit{hierarchy} of them \cite{BERKOVITZ2006661}. Near the top of this hierarchy sits the notion of Kolomogorov system (K-system) \cite{Walters_1982}, which is a system with complete memory loss. One of the main achievements of the field was to introduce the notion of Kolmogorov--Sinai (KS) entropy \cite{Sinai_2010}, and to show that under certain conditions, a system is a K-system if and only if it has completely positive Kolmogorov--Sinai entropy \cite{Rokhlin_1967,Cornfeld1982ErgodicT}. Therefore, this entropy constitutes a quantitative diagnostic of strong chaotic properties.

The ergodic theory of classical dynamical systems has a quantum counterpart, though it is much less developed. Some notions of chaos motivated by the classical hierarchy have been identified \cite{Benatti1993DeterministicCI}, however, an entropic characterization of the notion of quantum K-system remains elusive. In particular, the dynamical entropy defined by Connes--Narnhofer--Thirring (CNT entropy) \cite{cmp/1104160061} falls short of being able to fully characterize K-systems. It is still an open problem to obtain some quantitative diagnostic of the K-property in the quantum case. 

The goal of this note is to show that the quantum analog of the classical ergodic hierarchy is surprisingly relevant to the emergence of time in holography. In this hierarchy, the looser properties known as mixing properties encode the late time decay of correlation functions, while stronger properties like the K-property and the closely related Anosov property are closely tied to the emergence of horizons and information loss. Some ``intrinsically quantum" diagnostics, like the type of emergent von Neumann algebras \cite{Leutheusser:2021frk,Leutheusser:2021qhd,Leutheusser:2022bgi,Witten:2021unn,Chandrasekaran:2022eqq,Furuya:2023fei,Faulkner:2022ada,Gesteau:2023hbq}, also fit nicely into this hierarchy. By better understanding the parallels between the quantum ergodic hierarchy and the emergence of semiclassical physics, we may discover new ways to characterize how to interpolate between quantum gravity and its semiclassical limit. 

A theme of central importance will be that of von Neumann algebras, as it is the natural setting to formulate the quantum analog of measure theory, which lies at the heart of the theory of dynamical systems. Interestingly, the emergence of nontrivial types of von Neumann algebras has recently been identified as a diagnostic of the appearance of a somewhat semiclassical bulk \cite{Leutheusser:2021qhd,Leutheusser:2021frk,Leutheusser:2022bgi}, and one of the goals of this note is to sharpen this connection.

The aim of this note is to spell out a detailed comparison between the ergodic hierarchy and the emergence of semiclassical physics in quantum gravity. Rather than providing any novel computations, here results from disparate areas of mathematics and physics are assembled in order to draw in-depth parallels. The hope is that this may facilitate further exploration in this direction.

\textbf{Note added:} The material presented here has been expanded upon in \cite{Ouseph2023}, after the author explained it to one of the authors of \cite{Ouseph2023} during a conference this summer. We decided to coordinate our releases.

\section{Ergodic classical systems and their hierarchy}

The mathematical investigation of chaos in classical dynamical systems is well-developed, and known as ergodic theory. It was realized in the second half of the twentieth century that chaotic systems can be organized into a \textit{hierarchy} of systems satisfying stronger and stronger properties \cite{sep-ergodic-hierarchy}. This section briefly recalls these properties and gives some intuition on them. 

Before being able to introduce the ergodic hierarchy, it is important to set the stage and define what we precisely mean by a classical dynamical system. The idea of a dynamical system is that it describes evolution of the observables of a physical system in its space of configurations. This space $X$ is then equipped with a one-parameter group of automorphisms $\sigma_t$, which describes time evolution.

Two main possible choices for $X$ are possible: either $X$ is a topological space, in which case one talks about topological dynamics, or $X$ is a measurable space equipped with a measure $\mu$. For our purposes, it will be more useful to look at the latter case, and we will only briefly comment on the quantum equivalent of the former. The reason is that the notion of dynamical system relevant to study the semiclassical limit of holography should retain some notion of state-dependence, since it is the entanglement properties of a reference ``vacuum" state that are responsible for the emergence of spacetime. As we will see, the classical analog of a quantum state is a probability measure. It is then natural to appeal to measured spaces here. In contrast, topological dynamics is state-independent in nature, and less relevant to the context of this note.

\begin{defn}
A dynamical system is a triple $(X,\mu,\sigma_t)$, where $X$ is a measurable space, $\mu$ is a probability measure on $X$, and $\sigma_t$ is a measure-preserving flow on $X$.\footnote{The general discussion assumes that we are looking at systems with continuous time, but discrete time can be defined completely analogously.} 
\end{defn}

Here, we want to find a way to quantify how chaotic a dynamical system is. It turns out that there is a \textit{hierarchy} of chaos in the theory of classical dynamical systems, which consists of five increasingly strong levels of chaos. More precisely, there is the chain of implications\begin{widetext}
\begin{align*}
\textbf{Ergodic}\Leftarrow\textbf{Weakly Mixing}\Leftarrow\textbf{Strongly Mixing}\Leftarrow\textbf{Kolmogorov}\Leftarrow\textbf{Bernoulli}.
\end{align*}
\end{widetext}
The rest of the text will show that quantum analogs of most of the properties on this list, as well as the closely related Anosov systems, have been proposed to diagnose some emergent properties of the bulk in holographic theories.

\subsection{Ergodicity and mixing}

The weakest characterization of chaos is given by the notion of \textit{ergodicity}. The intuition behind this notion is that on average, the dynamics decorrelates events $A$ and $B$ and turns them into independent events. The weakest possible way of making sense of such a statement is given by the following definition:

\begin{defn}
A classical dynamical system $(X,\mu,\sigma_t)$ is \textit{ergodic} if for all measurable sets $A,B\subset X$,
\begin{align}
\underset{T\rightarrow\infty}{\mathrm{lim}}\frac{1}{2T}\int_{-T}^T\mu(\sigma_t(B)\cap A)dt=\mu(B)\mu(A).
\end{align}
\end{defn}

The first way in which one can strengthen the property of ergodicity is by imposing convergence to zero of the average \textit{difference} between the measure of the interesection of the two events and their product. It is a strengthening of the notion of convergence towards independence.

\begin{defn}
A classical dynamical system $(X,\mu,\sigma_t)$ is \textit{weakly mixing} if for all measurable sets $A,B\subset X$,
\begin{align}
\underset{T\rightarrow\infty}{\mathrm{lim}}\frac{1}{2T}\int_{-T}^T\lvert\mu(\sigma_t(B)\cap A)-\mu(B)\mu(A)\rvert dt=0.
\end{align}
\end{defn}

Of course a way of making the convergence even stronger is to ask for it not just on average but pointwise. This leads to the notion of strong mixing:

\begin{defn}
A classical dynamical system $(X,\mu,\sigma_t)$ is \textit{strongly mixing} if for all measurable sets $A,B\subset X$,
\begin{align}
\underset{t\rightarrow\infty}{\mathrm{lim}}\mu(\sigma_t(B)\cap A)=\mu(B)\mu(A).
\end{align}
\end{defn}

\subsection{Kolmogorov systems}

So far, the characterizations of chaos in dynamical systems have been fairly simple and quantitative: they simply are a matter of limits. However, one can ask for stronger properties of chaos, that are more algebraic in nature. In particular, another way of thinking about chaos is that dynamics should make an observer lose all possible information about the initial data. This is formalized by asking that there are less and less measurable sets left in the $\sigma$-algebra as time passes. The most extreme version of the idea leads to the notion of K-system (Kolmogorov system) \cite{Walters_1982}, which asks for complete memory loss at infinite time.

\begin{defn}
Let $(X,\mu,\sigma_t)$ be a classical dynamical system, and let $\Sigma$ be the $\sigma$-algebra of measurable subsets of $X$. $(X,\mu,\sigma_t)$ is a Kolmogorov system (K-system) if there exists a $\sigma$-algebra $\Sigma_0\subset\Sigma$ such that:
\begin{itemize}
\item For $t>0$, $\sigma_t(\Sigma_0)\subset\Sigma_0$,
\item $\bigvee_{t\in\mathbb{R}}\sigma_t(\Sigma_0)=\Sigma$,
\item $\bigwedge_{t\in\mathbb{R}}\sigma_t(\Sigma_0)=\{\emptyset,X\}$.
\end{itemize}
\end{defn}

There are other characterizations of K-systems in terms of things more akin to mixing properties. For example:

\begin{prop}[\cite{Cornfeld1982ErgodicT}]
Let $(X,\mu,\sigma_t)$ be a classical dynamical system on a Lebesgue space. $(X,\mu,\sigma_t)$ is a K-system if and only if it is K-mixing, i.e. for any measurable $A_0\subset X$, and any $r$ measurable subsets $A_1,\dots,A_r$ of $X$,
\begin{align}
\underset{t\rightarrow\infty}{\mathrm{lim}}\underset{B\in\Sigma_{t,r}}{\mathrm{sup}}\lvert\mu(B\cap A_0)-\mu(B)\mu(A_0)\rvert=0,
\end{align}
where $\Sigma_{t,r}$ is the $\sigma$-algebra generated by the $\sigma_s(A_i)$, for $s\geq t$.
\end{prop}

The notion of K-system is, in a lot of ways, the most interesting one of the whole hierarchy, and we will see that it is also the one whose quantum counterpart has the most striking physical interpretation. That is the reason why it is very important to come up with characterizations of the K-property that are more quantitative than its definition, so that they can actually be applied to concrete examples. In the classical case, this is achieved by constructing an invariant of classical dynamical systems, called Kolmogorov--Sinai (KS) entropy \cite{Sinai_2010}. The precise result is that a Lebesgue dynamical system has completely positive KS entropy if and only if it possesses has a K-structure \cite{Cornfeld1982ErgodicT}.

\subsection{Bernoulli systems}

The notion of Bernoulli shift sits at the very top of the ergodic hierarchy. The idea of such a system is that at each step, it completely forgets the previous one. A typical example of such a system is a repeated coin toss. More formally:

\begin{defn}
Let $(X,\mathcal{A},\mu)$ be a probability space. Then a Bernoulli scheme is $(X,\mathcal{A},\mu)^\mathbb{Z}$. A Bernoulli scheme equipped with its shift automorphism is called a Bernoulli dynamical system.
\end{defn}

The Bernoulli property is the strongest property of chaos one can ask for in a dynamical system, however, only few systems possess it. It will not be directly relevant in our case.

\subsection{Anosov systems}

Although not part of the canonical ergodic hierarchy, the notion of Anosov system is very important in the study of quantum dynamical systems. The idea of an Anosov flow requires a bit more structure, and it can be seen as a useful special case that often displays a lot of the properties described in this note. In the case of an Anosov system, the underlying measured space is actually a manifold (at least in the most basic case). The idea is then to split the tangent bundle of this manifold into contracting and expanding directions under the flow under consideration. More formally \cite{e345abdc-cb8a-3542-b921-1c0db6927d2b}:

\begin{defn}
Let $M$ be a smooth compact connected oriented manifold. Let $\phi_t$ be a nonsingular flow of class $C^r$ on $M$. $\phi_t$ is an Anosov flow of class $C^r$ if there exists a $\phi_t$-invariant continuous splitting of the tangent bundle of $M$, given by $TM=E^u\oplus E^s\oplus E^T$, where $E^T$ is the line bundle tangent to $\phi_t$ and $E^u$ and $E^s$ satisfy:
\begin{itemize}
\item There exist constants $A>0$, $\mu>1$ such that for $t\in\mathbb{R}$ and $v\in E^u$:\begin{align}\|\phi_{t\ast}(v)\|\geq A\mu^t\|v\|,\end{align}
\item There exist constants $B>0$, $\lambda<1$ such that for $t\in\mathbb{R}$ and $v\in E^s$:\begin{align}\|\phi_{t\ast}(v)\|\leq B\lambda^t\|v\|.\end{align}
\end{itemize}
\end{defn}

The splitting elements $E^u$ and $E^s$ can be interpreted as expanding and contracting directions for the flow $\phi_t$. The presence of these expansions and contractions is related to chaotic properties of the flow. This can be seen very explicitly in the case of geodesic flow on a Riemann surface, as it will be explained in Section \ref{sec:geodesic}.

Under reasonable conditions, the Anosov property implies the K-property. In particular,

\begin{thm}[\cite{BURNS2000149}]
Any $C^2$ topologically mixing volume preserving Anosov flow is a K-system (and even Bernoulli).
\end{thm}

It will turn out that in the case of quantum modular automorphisms, one obtains an even stronger link between the Anosov and K-properties, where to each K-system can be associated an Anosov flow.

\subsection{Example: Geodesic flow on a hyperbolic Riemann surface}
\label{sec:geodesic}
The prototypical example of an Anosov system is geodesic flow on a compact orientable Riemann surface of constant negative curvature. Any such surface can be written as a quotient $\Sigma=\Gamma\backslash\mathbb{H}$, where $\Gamma$ is a cocompact Fuchsian group. Consider $Q\cong PSL(2,\mathbb{R})$, the unit tangent bundle to the upper half-plane $\mathbb{H}$. The geodesic flow $\phi_t=\begin{pmatrix}e^{t/2}& 0 \\0 & e^{-t/2}\end{pmatrix}$, as well as the horocycle flows $h_t^\ast=\begin{pmatrix}1 & t \\0 & 1\end{pmatrix}$, $h_t=\begin{pmatrix}1 & 0 \\t & 1\end{pmatrix}$ define invariant flows on $Q$, which descend to invariant flows on $P$, the unit tangent bundle of $\Sigma$. Then one can decompose the tangent space $TP=E^+\oplus E^0\oplus E^-$, where $E^+$ and $E^-$ are the expanding and contracting directions along the horocycle flows, and $E^0$ is the geodesic flow direction. One can show that this decomposition satisfies the axioms of an Anosov system.

An important fact about this example is that the geodesic and horocycle flows satisfy the commutation relations
\begin{align}
g_th_s^\ast=h_{se^{-t}}^\ast g_t,\quad\quad\quad g_th_s=h_{se^t}g_t.
\label{eq:anosov}
\end{align}
These commutation relations will be the starting point of the definition of quantum Anosov systems. Note that they are the same as the ones between boosts and null translations in Rindler space, as will be explored in more detail in the next section.

\section{Quantum chaotic dynamics: von Neumann algebras}

One would like to define a similar hierarchy of chaos for the case of quantum systems. Of course, this implies resorting to a quantum notion of dynamical system on which measure theory can be used. The language of von Neumann algebras shows up naturally when one attempts to formulate a quantum analog of the notion of classical dynamical system. The reason is the following.

The space of configurations of a quantum system is of noncommutative nature. Therefore it is impossible to think of it as a classical topological manifold or as a measured space. However, there is a dual viewpoint on these spaces of configurations that generalizes well to the quantum case. The Gelfand--Naimark theorem \cite{Gelfand:1943} and its variants tell us that it is possible to think of various classes of classical spaces in terms of the algebras of functions they support. Points are then identified with ideals of functions that vanish on them. In the case of a topological space $X$, the relevant algebra is the algebra of continuous functions $C^0(X)$. In the case of a measured space, the relevant algebra is the algebra of bounded measurable functions $L^\infty(X,\mu)$.
It turns out that $C^0(X)$ is a commutative $C^\ast$-algebra, while $L^\infty(X,\mu)$ is a commutative von Neumann algebra.

Now focusing on the measured case, one needs to find a quantum analog of the notion of measure. Once again, since there no longer are any ``points" on a noncommutative configuration space, a measure itself cannot be defined. However, a similar notion still makes sense: that of taking the expectation value of a function against a probability measure thanks to the formula
\begin{align}
\braket{f}_\nu=\int_X f d\nu.
\end{align}

If $\nu$ is a probability measure, this is a linear functional of norm one on the space of functions. Thus it is natural to encode the notion of probability measure into that of expectation value functional in the noncommutative case - and an expectation value functional is nothing else than a quantum state. Hence, the quantum counterpart of the measure in a classical dynamical system is simply a quantum state.

In fact, one needs to be a bit more precise. In the classical case, if $f\in L^\infty(X,\mu)$ and $\nu$ is a measure on $X$, then $\braket{f}_\nu$ has good continuity properties with respect to $f$ if and only if $\nu$ is absolutely continuous with respect to $\mu$ (which is denoted by $\nu\ll\mu$). The quantum analog of this statement is that one should require the quantum state under consideration to be \textit{normal} on the von Neumann algebra.

Finally dynamics is trivially generalized as a (strongly continuous) one-parameter group of automorphisms of the von Neumann algebra $M$. This leads to the definition of quantum dynamical systems that will be used in this note:

\begin{defn}
A quantum dynamical system is a triple $(M,\omega,\tau_t)$, where $M$ is a von Neumann algebra, $\omega$ is a normal state on $M$, and $\tau_t$ is a strongly continuous one-parameter group of automorphisms of $M$.
\end{defn}
It is common to assume that $\omega\circ\tau_t=\omega$.

In the quantum case, a new situation of interest appears that used to be trivial in the commutative case: if $\omega$ is also faithful then $\tau_t$ can be the modular automorphism group of $\omega$. This is of particular interest given the central role played by modular flow in the emergence of spacetime. It is then useful to define:

\begin{defn}
If the quantum dynamical system $(M,\omega,\tau_t)$, with $\omega$ faithful, is defined so that $\tau_t$ is the modular automorphism group of $\omega$, then $(M,\omega,\tau_t)$ is said to be a modular quantum dynamical system.
\end{defn}

One can sum up the relationship between von Neumann algebras and measured spaces as indicated in Table \ref{tab:classquant}.

\begin{table}[!ht]
    \centering
    \begingroup
    \renewcommand{\arraystretch}{2.2}
    \scalebox{.8}{\begin{tabular}{cc}
    \toprule
\textbf{Classical case} & \textbf{Quantum case} \\
\midrule
$C^0(X)$ & $C^\ast$-algebra\\
\cmidrule(lr){1-2}
$L^\infty(X,\mu)$ & von Neumann algebra\\
\cmidrule(lr){1-2}
Measure $\nu\ll\mu$ & Normal state $\omega$ \\
\cmidrule(lr){1-2}
$\sigma$-algebra automorphism $\sigma_t$ & von Neumann algebra automorphism $\tau_t$\\
\cmidrule(lr){1-2}
Dynamical system $(X,\mu,\sigma_t)$ & Dynamical system $(M,\omega,\tau_t)$\\
\bottomrule
\end{tabular}}\endgroup
\caption{The parallel between classical and quantum dynamical systems.}
    \label{tab:classquant}
\end{table}

\subsection{Ergodicity and mixing}

The generalization of ergodicity and mixing to the quantum case is completely straightforward now that we have a good definition of a quantum dynamical system. We simply define:

\begin{defn}
The quantum dynamical system $(M,\omega,\tau_t)$ is:
\begin{itemize}
\item Ergodic if for all $A,B\in M$, 
\begin{align}
\underset{T\rightarrow\infty}{\mathrm{lim}}\frac{1}{2T}\int_{-T}^T\omega(\tau_t(B) A)=\omega(B)\omega(A),
\end{align}
\item Weakly mixing if for all $A,B\in M$, 
\begin{align}
\underset{T\rightarrow\infty}{\mathrm{lim}}\frac{1}{2T}\int_{-T}^T\lvert\omega(\tau_t(B) A)-\omega(B)\omega(A)\rvert dt=0,
\end{align}
\item Strongly mixing if for all $A,B\in M$, 
\begin{align}
\underset{t\rightarrow\infty}{\mathrm{lim}}\omega(\tau_t(B) A)=\omega(B)\omega(A).
\end{align}
\end{itemize}
\end{defn}

\subsection{Quantum Kolmogorov systems}

It is also quite straightforward to transfer the notion of K-system to the quantum case, and it stays close to the mixing properties.

\begin{defn}[\cite{cmp/1104179628}]
    Let $(M,\tau_t)$ be a quantum dynamical system. $M$ is \textit{refining} if there exists a von Neumann subalgebra $N\subsetneq M$ such that for $t>0$, $\sigma_t(N)\subset N$. It is a \textit{K-system} if in addition,
    $\underset{t\in\mathbb{R}}{\bigvee} \tau_t(N)=M$, and $\underset{t\in\mathbb{R}}{\bigcap} \tau_t(N)=\mathbb{C}$.    
\end{defn}

It is natural to ask whether just like KS entropy in the classical case, it is possible to find an invariant giving a characterization of K-flows in the quantum case. Unlike in the classical case, such an invariant has so far remained elusive. An important attempt to define a quantum analog of KS entropy has been made by Connes, Narnhofer and Thirring (CNT) \cite{cmp/1104160061}, however, it turns out that CNT entropy is not enough to characterize the K-property. Instead, an alternative notion of ``entropic K-system" (see for example \cite{Benatti1993DeterministicCI}) can be defined in analogy with the classical case, but it has not been shown to be equivalent to the algebraic definition. It is a very interesting question whether such difficulties could be overcome.

\subsection{Quantum Anosov systems}

There is also a good notion of Anosov flow for quantum systems, although one needs to relax it appropriately to allow for the contracting and expanding transformations to be non-invertible (only endomorphisms rather than automorphisms).

\begin{defn}[\cite{NARNHOFER_THIRRING_2000}]
Let $(M,\tau_t)$ be a quantum dynamical system, $\omega$ be an invariant faithful normal state on $M$, and $\mathcal{H}_\omega$ be its GNS representation. The system $(M,\tau_t)$ is \textit{Anosov} if there exists a strongly continuous one-parameter group of automorphisms $\sigma_s\in\mathrm{Aut}\mathcal{B}(\mathcal{H}_\omega)$ such that for $s>0$, $\sigma_s(M)\subsetneq M$, and for all $s,t\in\mathbb{R}$,
\begin{align}
\label{eq:rindlercomm}
\tau_t\circ\sigma_s=\sigma_{se^{-t}}\circ\tau_t.
\end{align}
\end{defn}

In the case of a modular quantum dynamical system, the Anosov and K-properties are very closely related:

\begin{prop}
Let $(M,\tau_t,\omega)$ be a modular quantum dynamical system. Then $(M,\tau_t,\omega)$ is refining if and only if it is Anosov.
\end{prop}
\begin{proof}
The proof closely follows \cite{NARNHOFER_THIRRING_2000}. For refining $\Rightarrow$ Anosov, let $M$ and $N$ be like in the definition. Let $\Delta_M$ and $\Delta_N$ implement the modular automorphisms of the restriction of $\omega$ to $M$ and $N$ respectively. As quadratic forms, $\Delta_M\leq\Delta_N$. Moreover, by the refining property, 
\begin{align}
\Delta_M^{it}\Delta_N^{-it}N\Delta_N^{it}\Delta_M^{-it}\subset N.
\end{align}
Let $W(t):=\Delta_M^{it}\Delta_N^{-it}$. $W(t)$ is an analytic function of $t$ on the strip $0<\mathrm{Im}\,t<1/2$. It is also strongly continuous at the boundary and unitary on the boundary, and we have the identity
\begin{align}
J_MW(s)J_N=W(i/2+s).
\end{align}
Let $H_{N,M}:=\frac{1}{2\pi}\mathrm{ln}\Delta_{N,M}$, and $G:=H_M-H_N\geq 0$.
The commutation relations are 
\begin{align}
[H_M,H_N]=i(H_M-H_N),
\end{align}
and
\begin{align}
[H_M,G]=iG.
\end{align}
Exponentiating these relations gives rise to the defining relation of a quantum Anosov system, where $H_M$ generates $\tau_t$ and $G$ generates $\sigma_s$. 

For Anosov $\Rightarrow$ refining, the algebra $N:=\sigma_1(M)$ satisfies the assumptions.
\end{proof}

Much more can be proven about modular Anosov/K-systems. Specifically, in \cite{NARNHOFER_THIRRING_2000} it is shown that under an extra assumption, any refining (or Anosov) system contains a K-subsystem. One can also show that the K-property and Anosov properties imply very strong mixing properties. For example, in \cite{doi:10.1142/S0219025701000401} it was shown that correlators decay exponentially for a dense set of observables of a quantum Anosov system, which is once again the expected behavior of a horizon in a quantum field theory \cite{Festuccia:2005pi,Festuccia:2006sa}, see also \cite{Narnhofer2007}. 

Finally note that in analogy to the classical case, there are some quantum K-systems that are not Bernoulli \cite{Golodets_Neshveyev_1998}. 

\subsection{Example: Arnold's cat and its quantization}

Another interesting example in this context is Arnold's cat map \cite{Emch:1994jt}, as it has both a classical and a quantum version. The principal difference from the other examples presented here is that this system has a discrete time parameter. In the classical case, the Arnold cat map is constructed in the following way:

\begin{defn}
Let $M:=\mathbb{R}^2/\mathbb{Z}^2$, and let $T\in SL(2,\mathbb{Z})$, with eigenvalues $\varepsilon_1$ and $\varepsilon_2$, and trace $>2$. Let $0<\varepsilon_1<1<\varepsilon_2$ be the eigenvalues of $T$, and $X_1$, $X_2$ be corresponding normalized eigenvectors.
\end{defn}

A standard result is:

\begin{thm}
The classical Arnold cat map defines an Anosov system.
\end{thm}

Now, one can quantize the cat map through usual Weyl quantization. Define the $C^\ast$-algebra $\mathcal{A}$ as the norm completion of the algebra generated by the normalized operators $W(x)$ for $x\in\mathbb{R}^2$, satisfying relations
\begin{align}
W(x)^\dagger=W(-x),
\end{align}
and
\begin{align}
W(x_1)W(x_2)=e^{i\pi \theta \mathrm{det}(x_1,x_2)}W(x_1+x_2),
\end{align}
restricted to the $x\in\mathbb{R}^2$ such that 
$W(\eta)W(x)W(\eta)^\dagger=W(x)$ for all $\eta\in(1/\theta)\mathbb{Z}$.
The state $\omega(W(\xi)):=\delta_{\xi,0}$ defines a faithful tracial state on $\mathcal{A}$. Let $M$ be the bicommutant of $\mathcal{A}$ in the GNS representation of $\omega$. The action $\tau:W(x)\mapsto W(^tTx)$ extends to a continuous action on $M$.

Defining the flow
\begin{align}
\sigma_t^{12}(W(x)):= e^{-it\braket{x,X_{12}}},
\end{align}
one can check that this extends to a weakly continuous one-parameter group of automorphisms on $M$, and that the quantum Anosov relation
\begin{align}
\tau^n\sigma_t^{12}\tau^{-n}=\sigma^{12}_{e^{-\lambda_{12}n}t}
\end{align}
holds. 

We then obtain:

\begin{thm}[\cite{Emch:1994jt}]
The quantized Arnold cat map defines a quantum Anosov system with discrete time evolution.
\end{thm}

Note that here the dynamics is discrete rather than continuous, nevertheless, the cat map is an important example of quantum Anosov system.

\subsection{Analogy with Rindler space}

As it was hinted above, Rindler space is the stereotypical example of a modular quantum K-system with continuous evolution. By the Bisognano--Wichmann theorem \cite{Bisognano:1976za}, modular flow of the algebra of observables in the Rindler wedge coincides with boosts along the Rindler horizon. If $\tau_t$ denotes modular flow and $\sigma^\pm_s$ denote the two null translations along the horizon, they then satisfy the relation 
\begin{align} \tau_t\circ \sigma^\pm_s\circ \tau_{-t}=\sigma^\pm_{se^{\pm t}},
\label{eq:rindler}\end{align} 
which matches with the Anosov relation \eqref{eq:rindlercomm}, and is part of the relations corresponding to the local Poincaré symmetry.

Moreover, one can note that denoting by $\mathcal{M}$ the algebra of the Rindler wedge, and by $\mathcal{N^\pm}:=\sigma^\pm_1(\mathcal{M})$, $\underset{t\in\mathbb{R}}{\bigcap} \tau_t(N)=\mathbb{C}$ \cite{lechner2022deformations}. In other words:

\begin{prop}[\cite{Emch:1994jt}]
The Rindler wedge is a quantum K-system with continuous dynamics. 
\end{prop}

What we learn from this example is that the structure of half-sided modular inclusion exactly fits into the definition of a modular K-system! Therefore there is an interpretation of the nesting of algebras along the horizon in terms of chaos: as one moves along the horizon there are less and less observables until there are none, which is a signature of complete memory loss. It is also interesting to note that the generators of null translations $G^\pm$ satisfy the commutation relations \begin{align}[H,G^\pm]=\pm i G^\pm,\end{align}
which were already related to some notion of maximal chaos in \cite{DeBoer:2019kdj}, where the $G^\pm$ referred to as ``modular scrambling modes".

Table \ref{tab:mostsymmetric} compares the properties of the Rindler wedge to those of geodesic flow on Riemann surfaces.

\begin{table}[!ht]
    \centering
    \begingroup
    \renewcommand{\arraystretch}{2.2}
    \scalebox{.8}{\begin{tabular}{cc}
    \toprule
\textbf{Geodesic flow on hyperbolic surfaces} & \textbf{QFT in Rindler space} \\
\midrule
Geodesic flow & Modular flow/boost\\
\cmidrule(lr){1-2}
Horocycle flow & Null translation\\
\cmidrule(lr){1-2}
Refining property & Half-sided modular inclusion \\
\cmidrule(lr){1-2}
K-property & Complete information loss\\
\bottomrule
\end{tabular}}\endgroup
\caption{The parallel between geodesic flow on hyperbolic surfaces and QFT in Rindler spacetime.}
    \label{tab:mostsymmetric}
\end{table}

\section{Application to holography}
\label{sec:holo}

This section uses the previously introduced formalism to summarize how the different elements of the ergodic hierarchy can be seen as various diagnostics of the emergence of semiclassical features of the bulk of a holographic theory. 

\subsection{Mixing and information loss}

The weakest properties of the ergodic hierarchy have to do with mixing, i.e. decay of two-point functions. In holography, the late-time decay of the two-point function at high temperature has been related to information loss at large $N$ \cite{Maldacena:2001kr,Festuccia:2005pi,Festuccia:2006sa,Furuya:2023fei}. It is also expected that the analytical properties of this large $N$ real time two-point functions give information on the black hole interior and the singularity. Therefore, one can interpret mixing as a signature of the emergence of at least some stringy notion of spacetime. On the other hand, the late time decay of two-point functions at high temperature is expected in weakly coupled gauge theories, whose holographic duals are very far from a usual local theory \cite{Festuccia:2005pi,Festuccia:2006sa}. Moreover, mixing on its own does not guarantee the emergence of a modular Anosov or K-structure. On the other hand, emergent half-sided modular inclusions are very related to these stronger properties.

\subsection{Emergent half-sided modular inclusions}

In \cite{Leutheusser:2021frk,Leutheusser:2021qhd,Leutheusser:2022bgi}, it was pointed out that the emergence of a half-sided modular inclusion in the large $N$ limit of a high temperature thermofield double state in AdS/CFT is related to the emergence of a black hole horizon and of time in the black hole interior, see Figure \ref{fig:HSMIAdS}. The commutation relation between null translations and boosts along the horizon in a half-sided modular inclusion is exactly \eqref{eq:rindler}, which corresponds to a modular Anosov structure, and allows to recover some of the Poincaré symmetry, following \cite{DeBoer:2019kdj}.

These emergent half-sided modular inclusions relate to the emergence of a thermodynamic arrow of time in the bulk. Recalling the interpretation of the nesting of the algebras inside a K/refining system, the fact that the algebra of observables becomes smaller and smaller can be interpreted as some emergent form of information loss.

\begin{figure}
\centering
\includegraphics[scale=.65]{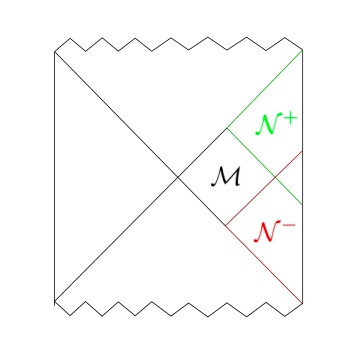}
\caption{The half-sided modular inclusions of Leutheusser--Liu, which are reinterpreted here as quantum Anosov systems. The algebras $\mathcal{N}^+$ and $\mathcal{N}^-$ represent the strict algebras of the right wedge algebra $\mathcal{M}$ of the thermofield double. They can be obtained from null translations along the horizon which, together with modular flow, give rise to an Anosov structure.}
\label{fig:HSMIAdS}
\end{figure}

\subsection{Remark on the role of type III$_1$}

In \cite{Leutheusser:2021frk,Leutheusser:2021qhd,Leutheusser:2022bgi}, it was pointed out that the emergence of type III$_1$ factors in the semiclassical limit of holography is a crucial feature in order to allow for the emergence of spacetime. What is now easy to note is that while type III$_1$ is necessary (mixing can only happen in a type III$_1$ factor \cite{Furuya:2023fei}), it may not be sufficient on its own to guarantee the emergence of a good notion of ``horizon" physics. Indeed, properties attached to half-sided modular inclusions, like the Anosov property and the K-property, are analogs of strictly stronger properties from the ergodic hierarchy.

\subsection{Some open questions}

It would be very interesting to reverse the logic of this note, and to ask whether given a holographic theory the $G_N\rightarrow 0$ limit, some notions of horizons, emergent times and local Poincaré symmetry appear in the bulk. It has been shown here that it may be possible to make progress on such a question by studying where the $G_N\rightarrow 0$ limit of the theory sits in the ergodic hierarchy. 

In the case of AdS/CFT, the relevant case is that of the large $N$ theory, which corresponds to generalized free fields. It was shown in \cite{Herman_Takesaki_1970} and reinterpreted in modern language in \cite{Furuya:2023fei} that the absolute continuity of the spectral density of these fields with respect to the Lebesgue measure guarantees that at finite temperature, the dynamics is strongly mixing at the level of the von Neumann algebra. In particular, this implies that the corresponding von Neumann algebra has type III$_1$. But what about the higher levels of the ergodic hierarchy introduced here? In particular:

\begin{itemize}
\item When is the K-property satisfied for generalized free fields at finite temperature?
\item Is there a quantitative characterization of the K-property in the quantum case, in analogy with complete positivity of KS entropy in the classical case?
\item To what extent can the K-property characterize the semiclassical nature of a bulk theory? 
\item Can the K-property be leveraged to obtain a notion of arrow of time, horizon physics and black hole interiors even when the bulk theory is subject to strong stringy effects?
\end{itemize}

\section*{Acknowledgements}

I am very grateful to Hong Liu for encouraging me to put this note together, to Matilde Marcolli for many related discussions, and to Allic Sivaramakrishnan for helpful comments on the draft.
\bibliographystyle{apsrev4-1}
\bibliography{refs,tfdmm}

\end{document}